\documentclass{article}
\usepackage{amsfonts}
\usepackage{graphicx}
\usepackage{amsmath,amssymb,amsthm}
\usepackage{authblk}
\usepackage{url}
\usepackage[linesnumbered,ruled]{algorithm2e}

\newcommand{\REV}[1]{\ensuremath{\overline{#1}}}

\newcommand{\RLBWT}{\ensuremath{\mathsf{RLBWT}}}
\newcommand{\ST}{\ensuremath{\mathsf{ST}}}

\newcommand{\BWT}{\ensuremath{\mathsf{BWT}}}
\newcommand{\CDAWG}{\ensuremath{\mathsf{CDAWG}}}
\newcommand{\MS}{\ensuremath{\mathsf{MS}}}
\newcommand\SP[1]{\mathtt{sp}(#1)}
\newcommand\EP[1]{\mathtt{ep}(#1)} 
\newcommand\INTERVAL[1]{\mathtt{range}(#1)}

\newcommand{\runs}{r}
\newcommand{\newe}{e}

\newtheorem{lemma}{Lemma} 
\newtheorem{theorem}{Theorem}
\newtheorem{property}{Property}
\newtheorem{corollary}{Corollary}

\title{Fast Label Extraction in the CDAWG}
\author[1]{Djamal Belazzougui}
\author[2]{Fabio Cunial}
\affil[1]{DTISI, CERIST Research Center, Algiers, Algeria.}
\affil[2]{Max Planck Institute of Molecular Cell Biology and Genetics, Dresden, Germany.}

\begin{document}
\maketitle

\begin{abstract}
The compact directed acyclic word graph (CDAWG) of a string $T$ of length $n$ takes space proportional just to the number $e$ of right extensions of the maximal repeats of $T$, and it is thus an appealing index for highly repetitive datasets, like collections of genomes from similar species, in which $e$ grows significantly more slowly than $n$. We reduce from $O(m\log{\log{n}})$ to $O(m)$ the time needed to count the number of occurrences of a pattern of length $m$, using an existing data structure that takes an amount of space proportional to the size of the CDAWG. This implies a reduction from $O(m\log{\log{n}}+\mathtt{occ})$ to $O(m+\mathtt{occ})$ in the time needed to locate all the $\mathtt{occ}$ occurrences of the pattern. We also reduce from $O(k\log{\log{n}})$ to $O(k)$ the time needed to read the $k$ characters of the label of an edge of the suffix tree of $T$, and we reduce from $O(m\log{\log{n}})$ to $O(m)$ the time needed to compute the matching statistics between a query of length $m$ and $T$, using an existing representation of the suffix tree based on the CDAWG. All such improvements derive from extracting the label of a vertex or of an arc of the CDAWG using a straight-line program induced by the reversed CDAWG.

\end{abstract}

\section{Introduction}

Large, highly repetitive datasets of strings are the hallmark of the post-genomic era, and locating and counting all the exact occurrences of a pattern in such collections has become a fundamental primitive. Given a string $T$ of length $n$, the compressed suffix tree \cite{RNOtalg11,NRdcc14.1} and the compressed suffix array can be used for such purpose, and they achieve an amount of space bounded by the $k$-th order empirical entropy of $T$. However, such measure of redundancy is known not to be meaningful when $T$ is very repetitive \cite{Gagie06a}. The space taken by such compressed data structures also includes an $o(n)$ term which can be a practical bottleneck when $T$ is very repetitive. Conversely, the size of the compact directed acyclic word graph (CDAWG) of $T$ is proportional just to the number of maximal repeats of $T$ and of their right extensions (defined in Section \ref{sec:strings}): this is a natural measure of redundancy for very repetitive strings, which grows sublinearly with $n$ in practice \cite{belazzougui2015composite}.

In previous work we described a data structure that takes an amount of space proportional to the size $e_T$ of the CDAWG of $T$, and that counts all the $\mathtt{occ}$ occurrences in $T$ of a pattern of length $m$ in $O(m\log{\log{n}})$ time, and reports all such occurrences in $O(m\log{\log{n}}+\mathtt{occ})$ time \cite{belazzougui2015composite}. We also described a representation of the suffix tree of $T$ that takes space proportional to the CDAWG of $T$, and that supports, among other operations, reading the $k$ characters of the label of an edge of the suffix tree in $O(k\log{\log{n}})$ time, and computing the matching statistics between a pattern of length $m$ and $T$ in $O(m\log{\log{n}})$ time. In this paper we remove the dependency of such key operations on the length $n$ of the uncompressed, highly repetitive string, without increasing the space taken by the corresponding data structures asymptotically. We achieve this by dropping the run-length-encoded representation of the Burrows-Wheeler transform of $T$, used in \cite{belazzougui2015composite}, and by exploiting the fact that the reversed CDAWG induces a context-free grammar that produces $T$ and only $T$, as described in \cite{belazzougui2017representing}. A related grammar, already implicit in \cite{crochemore2016linear}, has been concurrently exploited in \cite{takagi2017linear} to achieve similar bounds to ours. Note that in some strings, for example in the family $T_i$ for $i \geq 0$, where $T_0=0$ and $T_i=T_{i-1}iT_{i-1}$, the length of the string grows exponentially in the size of the CDAWG, thus shaving an $O(\log{\log{n}})$ term is identical to shaving an $O(\log{e_T})$ term. 

This work can be seen as a continuation of the research program, started in \cite{belazzougui2017representing,belazzougui2015composite}, of building a fully functional, repetition-aware representation of the suffix tree based on the CDAWG.

\section{Preliminaries}

We work in the RAM model with word length at least $\log{n}$ bits, where $n$ is the length of a string that is implicit from the context. We index strings and arrays starting from one. We call \emph{working space} the maximum amount of memory that an algorithm uses in addition to its input and its output.

\subsection{Graphs}

We assume the reader to be familiar with the notions of tree and of directed acyclic graph (DAG). In this paper we only deal with \emph{ordered} trees and DAGs, in which there is a total order among the out-neighbors of every node. The $i$-th leaf of a tree is its $i$-th leaf in depth-first order, and to every node $v$ of a tree we assign the compact integer interval $[\SP{v}..\EP{v}]$, in depth-first order, of all leaves that belong to the subtree rooted at $v$. In this paper we use the expression DAG also for directed acyclic \emph{multigraphs}, allowing distinct arcs to have the same source and destination nodes. In what follows we consider just DAGs with exactly one source and one sink. We denote by $\mathcal{T}(G)$ the tree generated by DAG $G$ with the following recursive procedure: the tree generated by the sink of $G$ consists of a single node; the tree generated by a node $v$ of $G$ that is not the sink, consists of a node whose children are the roots of the subtrees generated by the out-neighbors of $v$ in $G$, taken in order. Note that: (1) every node of $\mathcal{T}(G)$ is generated by exactly one node of $G$; (2) a node of $G$ different from the sink generates one or more internal nodes of $\mathcal{T}(G)$, and the subtrees of $\mathcal{T}(G)$ rooted at all such nodes are isomorphic; (3) the sink of $G$ can generate one or more leaves of $\mathcal{T}(G)$; (4) there is a bijection, between the set of root-to-leaf paths in $\mathcal{T}(G)$ and the set of source-to-sink paths in $G$, such that every path $v_1,\dots,v_k$ in $\mathcal{T}(G)$ is mapped to a path $v'_1,\dots,v'_k$ in $G$. 

\subsection{Strings} \label{sec:strings}

Let $\Sigma=[1..\sigma]$ be an integer alphabet, let $\#=0 \notin \Sigma$ be a separator, and let $T=[1..\sigma]^{n-1}\#$ be a string. Given a string $W \in [1..\sigma]^k$, we call the \emph{reverse of $W$} the string $\REV{W}$ obtained by reading $W$ from right to left. For a string $W \in [1..\sigma]^{k}\#$ we abuse notation, denoting by $\REV{W}$ the string $\REV{W[1..k]}\#$. Given a substring $W$ of $T$, let $\mathcal{P}_{T}(W)$ be the set of all starting positions of $W$ in 
$T$. A \emph{repeat} $W$ is a string that satisfies $|\mathcal{P}_{T}(W)|>1$. We conventionally assume that the empty string occurs $n+1$ times in $T$, before the first character of $T$ and after every character of $T$, thus it is a repeat. We denote by $\Sigma^{\ell}_{T}(W)$ the set of \emph{left extensions} of $W$, i.e. the set of characters $\{a \in [0..\sigma] : |\mathcal{P}_{T}(aW)|>0\}$. Symmetrically, we denote by $\Sigma^{r}_{T}(W)$ the set of \emph{right extensions} of $W$, i.e. the set of characters $\{b \in [0..\sigma] : |\mathcal{P}_{T}(Wb)|>0\}$. A repeat $W$ is \emph{right-maximal} (respectively, \emph{left-maximal}) iff $|\Sigma^{r}_{T}(W)|>1$ (respectively, iff $|\Sigma^{\ell}_{T}(W)|>1$). It is well known that $T$ can have at most $n-1$ right-maximal repeats and at most $n-1$ left-maximal repeats. A \emph{maximal repeat} of $T$ is a repeat that is both left- and right-maximal. Note that the empty string is a maximal repeat. A \emph{near-supermaximal repeat} is a maximal repeat with at least one occurrence that is not contained in an occurrence of another maximal repeat (see e.g. \cite{gusfield1997algorithms}). A \emph{minimal absent word} of $T$ is a string $W$ that does not occur in $T$, but such that any substring of $W$ occurs in $T$. It is well known that a minimal absent word $W$ can be written as $aVb$, where $a$ and $b$ are characters and $V$ is a maximal repeat of $T$ \cite{crochemore1998automata}. It is also well known that a maximal repeat $W=[1..\sigma]^m$ of $T$ is the equivalence class of all the right-maximal strings $\{W[1..m],\dots,W[k..m]\}$ such that $W[k+1..m]$ is left-maximal, and $W[i..m]$ is not left-maximal for all $i \in [2..k]$ (see e.g. \cite{belazzougui2015composite}). By \emph{matching statistics} of a string $S$ with respect to $T$, we denote the array $\MS_{S,T}[1..|S|]$ such that $\MS_{S,T}[i]$ is the length of the longest prefix of $S[i..|S|]$ that occurs in $T$.

For reasons of space we assume the reader to be familiar with the notion of \emph{suffix trie} of $T$, as well as with the related notion of \emph{suffix tree} $\ST_T=(V,E)$ of $T$, which we do not define here. We denote by $\ell(\gamma)$, or equivalently by $\ell(u,v)$, the label of edge $\gamma=(u,v) \in E$, and we denote by $\ell(v)$ the string label of node $v \in V$. It is well known that a substring $W$ of $T$ is right-maximal iff $W=\ell(v)$ for some internal node $v$ of the suffix tree. Note that the label of an edge of $\ST_T$ is itself a right-maximal substring of $T$, thus it is also the label of a node of $\ST_T$. We assume the reader to be familiar with the notion of \emph{suffix link} connecting a node $v$ with $\ell(v)=aW$ for some $a \in [0..\sigma]$ to a node $w$ with $\ell(w)=W$. Here we just recall that inverting the direction of all suffix links yields the so-called \emph{explicit Weiner links}. Given an internal node $v$ and a symbol $a \in [0..\sigma]$, it might happen that string $a\ell(v)$ does occur in $T$, but that it is not right-maximal, i.e. it is not the label of any internal node: all such left extensions of internal nodes that end in the middle of an edge or at a leaf are called \emph{implicit Weiner links}. The \emph{suffix-link tree} is the graph whose edges are the union of all explicit and implicit Weiner links, and whose nodes are all the internal nodes of $\ST_T$, as well as additional nodes corresponding to the destinations of implicit Weiner links. We call \emph{compact suffix-link tree} the subgraph of the suffix-link tree induced by maximal repeats. 

We assume the reader to be familiar with the notion and uses of the Burrows-Wheeler transform of $T$. In this paper we use $\BWT_T$ to denote the BWT of $T$, and we use $\INTERVAL{W} = [\SP{W}..\EP{W}]$ to denote the lexicographic interval of a string $W$ in a BWT that is implicit from the context. For a node $v$ (respectively, for an edge $e$) of $\ST_T$, we use the shortcut $\INTERVAL{v}=[\SP{v}..\EP{v}]$ (respectively, $\INTERVAL{e}=[\SP{e}..\EP{e}]$) to denote $\INTERVAL{\ell(v)}$ (respectively, $\INTERVAL{\ell(e)}$). We denote by $\runs_{T}$ the number of runs in $\BWT_{T}$, and we call \emph{run-length encoded BWT} (denoted by $\RLBWT_{T}$) any representation of $\BWT_{T}$ that takes $O(\runs_{T})$ words of space, and that supports rank and select operations (see e.g. \cite{makinen2005succinct1,MakinenNSV10,SirenVMN08}).

Finally, in this paper we consider only context-free grammars in which the right-hand side of every production rule consists either of a single terminal, or of at least two nonterminals. We denote by $\pi(F)$ the sequence of characters produced by a nonterminal $F$ of a context-free grammar. Every node in the parse tree of $F$ corresponds to an interval in $\pi(F)$. Given a nonterminal $F$ and an integer interval $[i..j] \subseteq [1..|\pi(F)|]$, let a node of the parse tree from $F$ be marked iff its interval is contained in $[i..j]$. By \emph{blanket} of $[i..j]$ in $F$ we denote the set of all marked nodes in the parse tree of $F$. Clearly the blanket of $[i..j]$ in $F$ contains $O(j-i)$ nodes and edges.

\subsection{CDAWG} \label{sec:cdawg}

The \emph{compact directed acyclic word graph} of a string $T$ (denoted by $\CDAWG_T$ in what follows) is the minimal compact automaton that recognizes all suffixes of $T$ \cite{blumer1987complete,CrochemoreV97}. We denote by $\newe_T$ the number of arcs in $\CDAWG_T$, and by $h_T$ the length of a longest path in $\CDAWG_T$. We remove subscripts when string $T$ is implicit from the context. The CDAWG of $T$ can be seen as the minimization of $\ST_T$, in which all leaves are merged to the same node (the sink) that represents $T$ itself, and in which all nodes except the sink are in one-to-one correspondence with the maximal repeats of $T$ \cite{Raffinot2001}. Every arc of $\CDAWG_T$ is labeled by a substring of $T$, and the out-neighbors $w_1,\dots,w_k$ of every node $v$ of $\CDAWG_T$ are sorted according to the lexicographic order of the distinct labels of arcs $(v,w_1),\dots,(v,w_k)$. We denote again with $\ell(v)$ (respectively, with $\ell(\gamma)$) the label of a node $v$ (respectively, of an arc $\gamma$) of $\CDAWG_T$. 

Since there is a bijection between the nodes of $\CDAWG_T$ and the maximal repeats of $T$, and since every maximal repeat of $T$ is the equivalence class of a set of roots of isomorphic subtrees of $\ST_T$, it follows that the node $v$ of $\CDAWG_T$ with $\ell(v)=W$ is the equivalence class of the nodes $\{v_1,\dots,v_k\}$ of $\ST_T$ such that $\ell(v_i)=W[i..m]$ for all $i \in [1..k]$, and such that $v_k,v_{k-1},\dots,v_1$ is a maximal unary path in the suffix-link tree. The subtrees of $\ST_T$ rooted at all such nodes are isomorphic, and $\mathcal{T}(\CDAWG_T)=\ST_T$. It follows that 
a right-maximal string can be identified by the maximal repeat $W$ it belongs to, and by the length of the corresponding suffix of $W$. Similarly, a suffix of $T$ can be identified by a length relative to the sink of $\CDAWG_T$.

The equivalence class of a maximal repeat is related to the equivalence classes of its in-neighbors in the CDAWG in a specific way:

\begin{property}[\cite{belazzougui2015composite}]\label{obs:inNeighbors}
Let $w$ be a node in the CDAWG with $\ell(w) = W \in [1..\sigma]^m$, and let $\mathcal{S}_w=\{W[1..m]$, $\dots$, $W[k..m]\}$ be the right-maximal strings that belong to the equivalence class of node $w$. Let $\{v^1,\dots,v^t\}$ be the in-neighbors of $w$ in $\CDAWG_T$, and let $\{V^1,\dots,V^t\}$ be their labels. Then, $\mathcal{S}_w$ is partitioned into $t$ disjoint sets $\mathcal{S}_w^1,\dots,\mathcal{S}_w^t$ such that $\mathcal{S}_w^i = \{W[x^i+1..m],W[x^i+2..m],\dots,W[x^i+|\mathcal{S}_{v^i}|..m]\}$, and the right-maximal string $V^{i}[p..|V^i|]$ labels the parent of the locus of the right-maximal string $W[x^i+p..m]$ in the suffix tree, for all $p \in [1..|\mathcal{S}_{v^i}|]$.
\end{property}

Property \ref{obs:inNeighbors} partitions every maximal repeat of $T$ into left-maximal factors, and applied to the sink $w$ of $\CDAWG_T$, it partitions $T$ into $t$ left-maximal factors, where $t$ is the number of in-neighbors of $w$, or equivalently the number of near-supermaximal repeats of $T$. Moreover, by Property \ref{obs:inNeighbors}, it is natural to say that in-neighbor $v^i$ of node $w$ is smaller than in-neighbor $v^j$ of node $w$ iff $x^i<x^j$, or equivalently if the strings in $\mathcal{S}^i_w$ are longer than the strings in $\mathcal{S}^j_w$. We call $\REV{\CDAWG}_T$ the ordered DAG obtained by applying this order to the reversed $\CDAWG_T$, i.e. to the DAG obtained by inverting the direction of all arcs of $\CDAWG_T$, and by labeling every arc $(v,w)$, where $w$ is the source of $\CDAWG_T$, with the first character of the string label of arc $(w,v)$ in $\CDAWG_T$. Note that some nodes of $\REV{\CDAWG}_T$ can have just one out-neighbor: for brevity we denote by $\REV{\CDAWG}_T$ the graph obtained by collapsing every such node $v$, i.e. by redirecting to the out-neighbor of $v$ all the arcs directed to $v$, propagating to such arcs the label of the out-neighbor of $v$, if any. 

The source of $\REV{\CDAWG}_T$ is the sink of $\CDAWG_T$, which is the equivalence class of all suffixes of $T$ in string order. There is a bijection between the distinct paths of $\REV{\CDAWG}_T$ and the suffixes of $T$; thus, the $i$-th leaf of $\mathcal{T}(\REV{\CDAWG}_T)$ in depth-first order corresponds to the $i$-th suffix of $T$ in string order. Moreover, the last arc in the source-to-sink path of $\REV{\CDAWG}_T$ that corresponds to suffix $T[i..|T|]$ is labeled by character $T[i]$. It follows that:

\begin{property}[\cite{belazzougui2017representing}] \label{obs:grammar}
$\REV{\CDAWG}_T$ is a context-free grammar that generates $T$ and only $T$, and $\mathcal{T}(\REV{\CDAWG}_T)$ is its parse tree. Let $v$ be a node of $\CDAWG_T$ with $t$ in-neighbors, and let $\ell(v)=VW$, where $W$ is the longest proper suffix of $\ell(v)$ that is a maximal repeat (if any). Then, $v$ corresponds to a nonterminal $F$ of the grammar such that $\pi(F) = V = \pi(F_1) \cdots \pi(F_t)$, and $F_i$ are the nonterminals that correspond to the in-neighbors of $v$, for all $i \in [1..t]$.
\end{property}

Note that the nonterminals of this grammar correspond to unary paths in the suffix-link tree of $T$, i.e. to edges in the suffix tree of $\REV{T}$. This parallels the grammar implicit in \cite{crochemore2016linear} and explicit in \cite{takagi2017linear}, whose nonterminals correspond to unary paths in the \emph{suffix trie} of $T$, i.e. to edges in the suffix tree of $T$.

\subsection{Counting and Locating with the CDAWG} \label{sec:locatingWithCDAWG}

$\CDAWG_T$ can be combined with $\RLBWT_T$ to build a data structure that takes $O(\newe_{T})$ words of space, and that counts all the $\mathtt{occ}$ occurrences of a pattern $P$ of length $m$ in $O(m\log{\log{n}})$ time, and reports all such occurrences in $O(m\log{\log{n}} + \mathtt{occ})$ time \cite{belazzougui2015composite}. 

Specifically, for every node $v$ of the CDAWG, we store $|\ell(v)|$ in a variable $v.\mathtt{length}$. Recall that an arc $(v,w)$ in the CDAWG means that maximal repeat $\ell(w)$ can be obtained by extending maximal
repeat $\ell(v)$ to the right \emph{and to the left}. Thus, for every arc $\gamma=(v,w)$ of the CDAWG, we store the first character of $\ell(\gamma)$ in a variable $\gamma.\mathtt{char}$, and we store the length of the right extension implied by $\gamma$ in a variable $\gamma.\mathtt{right}$. The length $\gamma.\mathtt{left}$ of the left extension implied by $\gamma$ can be computed by $w.\mathtt{length}-v.\mathtt{length}-\gamma.\mathtt{right}$. For every arc of the CDAWG that connects a maximal repeat $W$ to the sink, we store just $\gamma.\mathtt{char}$ and the starting position $\gamma.\mathtt{pos}$ of string $W \cdot \gamma.\mathtt{char}$ in $T$. The total space used by the CDAWG is $O(e_{T})$ words, and the number of runs in $\BWT_{T}$ can be shown to be $O(e_{T})$ as well \cite{belazzougui2015composite}.

We use the RLBWT to count the number of occurrences of $P$ in $T$, in $O(m\log{\log{n}})$ time: if this number is not zero, we use the CDAWG to report all the $\mathtt{occ}$ occurrences of $P$ in $O(\mathtt{occ})$ time, using a technique already sketched in
\cite{crochemore1997automata}. Specifically, since we know that $P$ occurs in $T$, we perform a \emph{blind search} for $P$ in the CDAWG, as follows. We keep a variable $i$, initialized to zero, that stores the length of the prefix of $P$ that we have matched so far, and we keep a variable $j$, initialized to one, that stores the starting position of $P$ inside the last maximal repeat encountered during the search. For every node $v$ in the CDAWG, we choose the arc $\gamma$ such that $\gamma.\mathtt{char}=P[i+1]$ in constant time using hashing, we increment $i$ by $\gamma.\mathtt{right}$, and we increment $j$ by $\gamma.\mathtt{left}$. If the search leads to the sink by an arc $\gamma$, we report $\gamma.\mathtt{pos}+j-1$ and we stop. If the search ends at a node $v$ that is associated with a maximal repeat $W$, we determine all the occurrences of $W$ in $T$ by performing a depth-first traversal of all nodes reachable from $v$ in the CDAWG, updating variables $i$ and $j$ as described above, and reporting $\gamma.\mathtt{pos}+j-1$ for every arc $\gamma$ that leads to the sink. Clearly the total number of nodes and arcs reachable from $v$ is $O(\mathtt{occ})$.

Note that performing the blind search for a pattern in the CDAWG is analogous to a descending walk on the suffix tree, thus we can compute the BWT interval of every node of $\ST_T$ that we meet during the search, by storing in every arc of the CDAWG a suitable offset between BWT intervals, as described in the following property:

\begin{property}[\cite{belazzougui2015composite}]\label{obs:bwtIntervals}
Let $\{W[1..m],\dots,W[k..m]\}$ be the right-maximal strings that belong to the equivalence class of maximal repeat $W \in [1..\sigma]^m$ of string $T$, and let $\INTERVAL{W[i..m]}=[p_i..q_i]$ for $i \in [1..k]$. Then $|q_i-p_i+1|=|q_j-p_j+1|$ for all $i$ and $j$ in $[1..k]$. Let $c \in [0..\sigma]$, and let $\INTERVAL{W[i..m]c}=[x_i..y_i]$ for $i \in [1..k]$. Then, $x_i = p_i+x_1-p_1$ and $y_i = p_i+y_1-p_1$.
\end{property}

Properties \ref{obs:inNeighbors} and \ref{obs:bwtIntervals}, among others, can be used to implement a number of suffix tree operations in $O(1)$ or $O(\log{\log{n}})$ time, using data structures that take just $O(\newe_T)$ or $O(\newe_T+\newe_{\REV{T}})$ words of space \cite{belazzougui2017representing,belazzougui2015composite}. Among other information, such data structures store a pointer, from each node $v$ of the CDAWG, to the longest proper suffix of $\ell(v)$ (if any) that is a maximal repeat. Note that such \emph{suffix pointers} can be charged to suffix links in $\ST_T$, thus they take overall $O(\newe_T)$ words of space.

\section{Faster Count and Locate Queries in the CDAWG} \label{sec:fasterCount}

\begin{algorithm}[t]
\SetCommentSty{normal}
{$S \gets \mbox{empty stack}$}\;
{$S.\mathtt{push}( (u',0,0) )$}\;
{$\mathtt{extracted} \gets 0$}\;
\Repeat{$\mathtt{extracted}=k$}{
   {$t \gets S.\mathtt{top}$}\;
   \uIf{$t.\mathtt{lastChild}<|t.\mathtt{node}.\mathtt{outNeighbors}|$}{
      {$t.\mathtt{lastChild} \gets t.\mathtt{lastChild}+1$}  \label{line:pointer}\;
      {$v' \gets t.\mathtt{node}.\mathtt{outNeighbors}[t.\mathtt{lastChild}]$}\;
      \uIf{$v'=G.\mathtt{sink}$}{
         {$\mathtt{print}(\mathtt{label}(t.\mathtt{node},v'))$}\;
         {$\mathtt{extracted} \gets \mathtt{extracted}+1$}\;
      }
      \uElseIf{$t.\mathtt{lastChild}=1$}{
         {$t.\mathtt{depth} \gets 1$}\;
         {$S.\mathtt{push}( (\mathtt{levelAncestor}(t.\mathtt{node},t.\mathtt{depth}),0,t.\mathtt{depth}) )$}\;
      }
      \lElse{
         {$S.\mathtt{push}( (v',0,0) )$}
      }
   }
   \Else{
      {$S.\mathtt{pop}$}\;
      \lIf{$S=\varnothing$}{
         \Return{$\mathtt{extracted}$}
      }
      {$t \gets S.\mathtt{top}$}\;
      \lIf{$t.\mathtt{depth}<t.\mathtt{node}.\mathtt{depth}$}{
         {$t.\mathtt{depth} \gets t.\mathtt{depth}+1$}
      }
      \lIf{$t.\mathtt{depth}<t.\mathtt{node}.\mathtt{depth}$}{
         {$S.\mathtt{push}( (\mathtt{levelAncestor}(t.\mathtt{node},t.\mathtt{depth}),1,t.\mathtt{depth}) )$}
      }
   }
}
\Return{$k$}\;
\caption{\label{algo:scan} Reading the first $k$ characters of the string produced by a nonterminal $F$ of a straight-line program represented as a DAG $G$. $F$ corresponds to node $u'$ of $G$. Notation follows Lemma \ref{lemma:scan}. 
}
\end{algorithm}

In this paper we focus on deciding whether a pattern $P$ occurs in $T$, a key step in the blind search of Section \ref{sec:locatingWithCDAWG}. Rather than using the RLBWT for such decision, we exploit Property \ref{obs:grammar} and use the grammar induced by $\REV{\CDAWG}_T$.

Our methods will require a data structure, of size linear in the grammar, that extracts in $O(k)$ time the first $k$ characters of the string produced by a nonterminal. Previous research described an algorithm that extracts the whole string produced by a nonterminal in linear time, using just constant working space, by manipulating pointers in the grammar \cite{gasieniec2003time}. This solution does not guarantee linear time when just a prefix of the string is extracted. A linear-size data structure with the stronger guarantee of constant-time extraction per character has also been described \cite{gasieniec2005real}, and this solution can be used as a black box in our methods. However, since we just need amortized linear time, we describe a significantly simpler alternative that needs just a level ancestor data structure (an idea already implicit in~\cite{lohrey2016traversing}) and that will be useful in what follows:

\begin{lemma} \label{lemma:scan}
Let $G=(V,E)$ be the DAG representation of a straight-line program. There is a data structure that: (1) given an integer $k$ and a nonterminal $F$, allows one to read the first $k$ characters of $\pi(F)$ in $O(k)$ time and $O(\min\{k,h\})$ words of working space, where $h$ is the height of the parse tree of $F$; 
(2) given a string $S$ and a nonterminal $F$, allows one to compute the length $k$ of the longest prefix of $S$ that matches a prefix of $\pi(F)$, in $O(k)$ time and $O(\min\{k,h\})$ words of working space. Such data structure takes $O(|V|)$ words of space.
\end{lemma}
\begin{proof}
We mark the arc of $G$ that connects each node $v'$ to its first out-neighbor. The set of all marked arcs induces a spanning tree $\tau$ of $G$, rooted at the sink and arbitrarily ordered \cite{gasieniec2005real}. 
In what follows we identify the nodes of $\tau$ with the corresponding nodes of $G$. We build a data structure that supports \emph{level ancestor queries} on $\tau$: given a node $v$ and an integer $d$, such data structure returns the ancestor $u$ of $v$ in $\tau$ such that the path from the root of $\tau$ to $u$ contains exactly $d$ edges. The level ancestor data structure described in \cite{bender2004level,berkman1994finding} takes $O(|V|)$ words of space and it answers queries in constant time.
To read the first $k$ characters of string $\pi(F)=W$, we explore the blanket of $W[1..k]$ in $F$ recursively, as described in Algorithm \ref{algo:scan}. The tuples in the stack used by the algorithm have the following fields: $(\mathtt{node},\mathtt{lastChild},\mathtt{depth})$, where $\mathtt{node}$ is a node of $G$, $u'.\mathtt{outNeighbors}$ is the sorted list of out-neighbors of node $u'$ in $G$, $u'.\mathtt{depth}$ is the depth of $u'$ in $\tau$, and function $\mathtt{label}(u',v')$ returns the character that labels arc $(u',v')$ in $G$. Algorithm \ref{algo:scan} returns the number of characters read, which might be smaller than $k$.
A similar procedure can be used for computing the length of the longest prefix of $\pi(F)$ that matches a prefix of a query string. Every type of operation in Algorithm \ref{algo:scan} takes constant time, it can be charged to a distinct character in the output, and it pushes at most one element on the stack. Thus, the stack contains $O(k)$ tuples at every step of the algorithm. It is also easy to see that the stack never contains more elements than the length of the longest path from the node of $G$ that corresponds to $F$ to the sink. 
\end{proof}

If necessary, Algorithm \ref{algo:scan} can be modified to take constant time per character:

\begin{corollary} \label{cor:scan}
Let $G=(V,E)$ be the DAG representation of a straight-line program. There is a data structure that takes $O(|V|)$ words of space and that, given a nonterminal $F$, allows one to read the characters of $\pi(F)$, from left to right, in constant time per character and in $O(\min\{k,h\})$ words of working space, where $h$ is the height of the parse tree of $F$.
\end{corollary}
\begin{proof}
After having printed character $i$ of $\pi(F)$, the time Algorithm \ref{algo:scan} has to wait before printing character $i+1$ is always bounded by a constant, except when the procedure repeatedly pops tuples from the stack. This can be avoided by preventively popping a tuple $t$ for which $t.\mathtt{lastChild}$ has reached $|t.\mathtt{node}.\mathtt{outNeighbors}|$ after Line \ref{line:pointer} is executed, before pushing new tuples on the stack. 
\end{proof}

Moreover, Lemma \ref{lemma:scan} can be generalized to weighted DAGs, by storing in each node of $\tau$ the sum of weights of all edges from the node to the root of $\tau$, by saving sums of weights in the tuples on the stack, and by summing and subtracting the weights of the arcs of the DAG:

\begin{corollary} \label{cor:generalizedScan}
Let $G=(V,E)$ be an ordered DAG with a single sink and with weights on the arcs, and let the \emph{weight of a path} be the sum of weights of all its arcs. There is a data structure that, given an integer $k$ and a node $v$, reports the weights of the first $k$ paths from $v$ to the sink in preorder, in constant time per path and in $O(\min\{k,h\})$ words of working space, where $h$ is the length of a longest path from $v$ to the sink. Such data structure takes $O(|V|)$ words of space.
\end{corollary}

Lemma \ref{lemma:scan} is all we need to verify in linear time whether a pattern occurs in the indexed text:

\begin{theorem} \label{thm:verification}
Let $T \in [1..\sigma]^n$ be a string. There is a data structure that takes $O(\newe_T)$ words of space, and that counts (respectively, reports) all the $\mathtt{occ}$ occurrences of a pattern $P \in [1..\sigma]^m$ in $O(m)$ time (respectively, in $O(m+\mathtt{occ})$ time) and in $O(\min\{m,h_T\})$ words of working space.
\end{theorem}
\begin{proof}
We assume that every node $v'$ of $\CDAWG_T$ stores in a variable $v'.\mathtt{freq}$ the number of occurrences of $\ell(v')$ in $T$. Recall that, for a node $v'$ of $\CDAWG_T$, $\ell(v') = \pi(F_1) \pi(F_2) \cdots \pi(F_k) \cdot W$, where $F_p$ for $p \in [1..k]$ are nonterminals of the grammar, and $W$ is the maximal repeat that labels the node $w'$ of $\CDAWG_T$ that is reachable from $v'$ by a suffix pointer. For each arc $(u',v')$ of $\CDAWG_T$, we store a pointer to the nonterminal $F_p$ of $v'$ that corresponds to $u'$. We perform a blind search for $P$ in $\CDAWG_T$ as described in Section \ref{sec:locatingWithCDAWG}: either the search is unsuccessful, or it returns a node $v'$ of $\CDAWG_T$ and an integer interval $[i..j]$ such that, if $P$ occurs in $T$, then $P=V[i..j]$ where $V=\ell(v')$, and the number of occurrences of $P$ in $T$ is $v'.\mathtt{freq}$. To decide whether $P$ occurs in $T$, we reconstruct the characters in $V[i..j]$ as follows (Figure \ref{fig:arc}a). Clearly $i$ belongs to a $\pi(F_p)$ for some $p$, and such $F_p$ can be accessed in constant time using the pointers described at the beginning of the proof. If $i$ is the first position of $\pi(F_p)$, we extract all characters of $\pi(F_p)$ by performing a linear-time traversal of the parse tree of $F_p$. Otherwise, we extract the suffix of $\pi(F_p)$ in linear time using Lemma \ref{lemma:scan}. Note that $j$ must belong to $\pi(F_q)$ for some $q>p$, since the search reaches $v'$ after right-extending a suffix of an in-neighbor $u'$ of $v'$ that belongs to the equivalence class of $u'$ (recall Property \ref{obs:inNeighbors}). We thus proceed symmetrically, traversing the entire parse tree of $F_{p+1} \dots F_{q-1}$ and finally extracting either the entire $\pi(F_q)$ or a prefix. Finally, $j$ could belong to $W$, in which case we traverse the entire parse tree of $F_{p+1} \dots F_k$ and we recur on $w'$, resetting $j$ to $j-\sum_{x=1}^{k}|\pi(F_x)|$. If the verification is successful, we proceed to locate all the occurrences of $P$ in $T$ as described in Section \ref{sec:locatingWithCDAWG}. 
\end{proof}

Note that the data structure in Theorem \ref{thm:verification} takes actually $O(\min\{\newe_T,\newe_{\REV{T}}\})$ words of space, since one could index either $T$ or $\REV{T}$ for counting and locating. Lemma \ref{lemma:scan} can also be used to report the top $k$ occurrences of a pattern $P$ in $T$, according to the popularity of the right-extensions of $P$ in the corpus:

\begin{corollary} \label{cor:topK}
Let $P$ be a pattern, let $\mathcal{P}=\{p_1,p_2,\dots,p_m\}$ be the set of all its starting positions in a text $T$. Let sequence $Q = q_1,q_2,\dots,q_m$ be such that $q_i \in \mathcal{P}$ for all $i \in [1..m]$, $q_i \neq q_j$ for all $i \neq j$, and $i<j$ iff $T[q_i..|T|]$ is lexicographically smaller than $T[q_j..|T|]$. Let sequence $S = s_1,s_2,\dots,s_m$ be such that $s_i \in \mathcal{P}$ for all $i \in [1..m]$, $s_i \neq s_j$ for all $i \neq j$, and $i<j$ iff the frequency of $T[s_i..s_i+x]$ in $T$ is not smaller than the frequency of $T[s_j..s_j+x]$ in $T$ (with ties broken lexicographically), where $x$ is the length of the longest common prefix between $T[s_i..|T|]$ and $T[s_j..|T|]$. There is a data structure that allows one to return the first $k$ elements of sequence $Q$ or $S$ in constant time per element and in $O(\min\{k,h_T\})$ words of working space. Such data structure takes $O(\newe_T)$ words of space.
\end{corollary}
\begin{proof}
Recall that Theorem \ref{thm:verification} builds the spanning tree $\tau$ of Lemma \ref{lemma:scan} on the reversed CDAWG that represents a straight-line program of $T$. To print $Q$, we build $\tau$ and the corresponding level-ancestor data structure on $\CDAWG_T$, connecting each vertex of the CDAWG to its lexicographically smallest out-neighbor, and storing in each node of $\tau$ the sum of lengths of all edges from the node to the root of $\tau$. Given the locus $v'$ of $P$ in $\CDAWG_T$, we can print the first $k$ elements of $Q$ in $O(k)$ time and in $O(\min\{k,h\})$ words of space, where $h$ is the length of a longest path from $v'$ to the sink of $\CDAWG_T$, by using Corollary~\ref{cor:generalizedScan}. To print $S$ we add to each node of $\CDAWG_T$ an additional list of children, sorted by nondecreasing frequency with ties broken lexicographically, and we build the spanning tree $\tau$ by connecting each vertex of $\CDAWG_T$ to its first out-neighbor in such new list. 
\end{proof}

Finally, Theorem \ref{thm:verification} allows one to reconstruct the label of any arc of the CDAWG, in linear time in the length $k$ of such label. This improves the $O(k\log{\log{n}})$ bound described in \cite{belazzougui2015composite}, where $n$ is the length of the uncompressed text, and it removes the $\newe_{\REV{T}}$ term from the space complexity, since $\RLBWT_{\REV{T}}$ is not needed.

\begin{figure}[t!]
\centering
\includegraphics[width=1.0\textwidth]{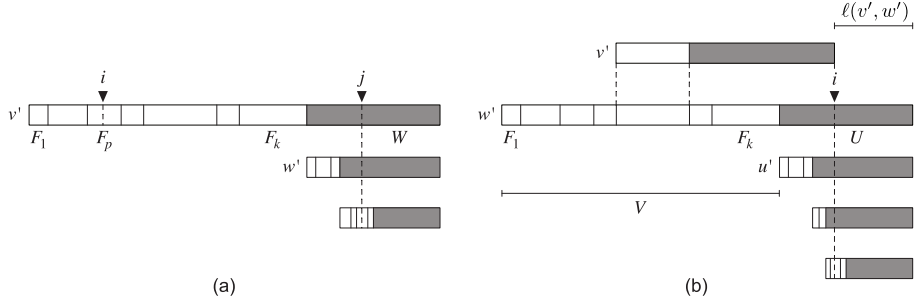}
\caption{(a) The verification step of pattern search, implemented with the CDAWG. Notation follows Theorem \ref{thm:verification}. (b) Reconstructing the label of an arc of the CDAWG. Notation follows Theorem \ref{thm:arcLabel}.}
\label{fig:arc}
\end{figure}

\begin{theorem} \label{thm:arcLabel}
There is a data structure that allows one to read the $k$ characters of the label of an arc $(v',w')$ of $\CDAWG_T$, in $O(k)$ time and in $O(\min\{k,h_T\})$ words of working space. Such data structure takes $O(\newe_T)$ words of space.
\end{theorem}
\begin{proof}
Recall that every arc $(v',w')$ that does not point to the sink of $\CDAWG_T$ is a right-maximal substring of $T$. If it is also a maximal repeat, then we can already reconstruct it as described in Theorem \ref{thm:verification}, storing a pointer to such maximal repeat, starting extraction from the first nonterminal of the maximal repeat, and recurring to the maximal repeat reachable from its suffix pointer. Otherwise, let $W=\ell(w')=VU$, where $U$ is the maximal repeat that corresponds to the node $u'$ reachable from the suffix pointer of $w'$, and let $V = \pi(F_1) \cdots \pi(F_k)$ where $F_p$ for $p \in [1..k]$ are nonterminals in the grammar. The label of $(v',w')$ coincides with suffix $W[i..|W|]$, and its length is stored in the index.

If $i \leq |V|$, let $V[i..|V|] = X \cdot \pi(F_{p+1}) \cdots \pi(F_k)$ for some $p$. To reconstruct $U$, we traverse the whole parse tree of $F_k,F_{k-1},\dots,F_{p+1}$, and we reconstruct the suffix of length $|X|$ of $\pi(F_p)$ using Lemma \ref{lemma:scan}. Otherwise, if $i>|V|$, we could recur to $U$, resetting $i$ to $i-|V|$ (Figure \ref{fig:arc}b). Let $U=V'U'$, where $U'$ is the maximal repeat that corresponds to the node reachable from the suffix pointer of $u'$. Note that it could still happen that $i>|V'|$, thus we might need to follow a sequence of suffix pointers. During the construction of the index, we store with arc $(v',w')$ a pointer to the first maximal repeat $t'$, in the sequence of suffix pointers from $w'$, such that $|\ell(t')| \geq |\ell(v',w')|$, and such that the length of the longest proper suffix of $\ell(t')$ that is a maximal repeat is either zero or smaller than $|\ell(v',w')|$. To reconstruct $\ell(v',w')$, we just follow such pointer and proceed as described above.

Reading the label of an arc that is directed to the sink of $\CDAWG_T$ can be implemented in a similar way: we leave the details to the reader. 
\end{proof}

We can also read the label of an arc $(v',w')$ \emph{from right to left}, with the stronger guarantee of taking constant time per character:

\begin{corollary} \label{cor:rightToLeft}
There is a data structure that allows one to read the $k$ characters of the label of an arc $(v',w')$ of $\CDAWG_T$, from right to left, in constant time per character and in $O(\min\{k,h_T\})$ words of working space. Such data structure takes $O(\newe_T)$ words of space.
\end{corollary}
\begin{proof}
We proceed as in Theorem \ref{thm:arcLabel}, but we also keep the tree $\tau$ of explicit Weiner links from every node of $\CDAWG_T$, imposing an arbitrary order on the children of every node $t$ of $\tau$, and we build a data structure that supports level ancestor queries on $\tau$. 
As in Theorem \ref{thm:arcLabel}, we move to a maximal repeat $u'$ such that $|\ell(u')| \geq |\ell(v',w')|$, and such that the length of the longest proper suffix of $\ell(u')$ that is a maximal repeat is either zero or smaller than $|\ell(v',w')|$. Then, we move to node $x'=\mathtt{levelAncestor}(u',1)$, we reconstruct $\ell(x')$ from right to left using Corollary \ref{cor:scan}, and we use $\mathtt{levelAncestor}(u',2)$ to follow an explicit Weiner link from $x'$. After a sequence of such explicit Weiner links we are back to $u'$, and we reconstruct from right to left the prefix of $\ell(v',w')$ that does not belong to the longest suffix of $\ell(u')$ that is a maximal repeat, using again Corollary \ref{cor:scan}. 
\end{proof}

Since the label of arc $(v',w')$ is a suffix of $\ell(w')$, and since the label of every node $w'$ of the CDAWG can be represented as $\pi(F) \cdot \ell(u')$, where $F$ is a nonterminal of the grammar and $u'$ is the longest suffix of $\ell(w')$ that is a maximal repeat, we could implement Corollary \ref{cor:rightToLeft} by adding to the grammar the nonterminals $W'$ and $U'$ and a new production $W' \rightarrow F U'$ for nodes $w'$ and $u'$, and by using Corollary \ref{cor:scan} for extraction. This does not increase the size of the grammar asymptotically. Note that the subgraph induced by the new nonterminals in the modified grammar is the reverse of the compact suffix-link tree of $T$.

\section{Faster Matching Statistics in the CDAWG} \label{sec:fasterMS}

A number of applications, including matching statistics, require reading the label of an arc \emph{from left to right}: this is not straightforward using the techniques we described, since the label of an arc $(v',w')$ can start e.g. in the middle of one of the nonterminals of $w'$ rather than at the beginning of one such nonterminal (see Figure \ref{fig:arc}b). We circumvent the need for reading the characters of the label of an arc from left to right in matching statistics, by applying the algorithm in Theorem \ref{thm:verification} to prefixes of the pattern of exponentially increasing length:

\begin{lemma} \label{lemma:leftToRight}
There is a data structure that, given a string $S$ and an arc $(v',w')$ of $\CDAWG_T$, allows one to compute the length $k$ of the longest prefix of $S$ that matches a prefix of the label of $(v',w')$, in $O(k)$ time and in $O(\min\{k,h_T\})$ words of working space. Such data structure takes $O(\newe_T)$ words of space.
\end{lemma}
\begin{proof}
Let $\gamma=(v',w')$. If $\ell(\gamma)$ is a maximal repeat of $T$, we can already read its characters from left to right by applying Theorem \ref{thm:verification}. Otherwise, we perform a doubling search over the prefixes of $S$, testing iteratively whether $S[1..2^i]$ matches a prefix of $\ell(\gamma)$ for increasing integers $i$, and stopping when $S[1..2^i]$ does not match a prefix of $\ell(\gamma)$. We perform a linear amount of work in the length of each prefix, thus a linear amount of total work in the length of the longest prefix of $S$ that matches a prefix of $\ell(\gamma)$.

We determine whether $S[1..2^i]$ is a prefix of $\ell(\gamma)$ as follows. Recall that an arc of $\CDAWG_T$ (or equivalently of $\ST_T$) is a right-maximal substring of $T$, therefore it is also a node of $\ST_T$. We store for each arc $\gamma$ of $\CDAWG_T$ the interval $\INTERVAL{\gamma}$ of the corresponding string in $\BWT_T$. Given $S[1..2^i]$, we perform a blind search on the CDAWG, simulating a blind search on $\ST_T$ and using Property \ref{obs:bwtIntervals} to keep the BWT intervals of the corresponding nodes of $\ST_T$ that we meet. We stop at the node $v$ of the suffix tree at which the blind search fails, or at the first node whose interval does not contain $\INTERVAL{\gamma}$ (in which case we reset $v$ to its parent), or at the last node reached by a successful blind search in which the BWT intervals of all traversed nodes contain $\INTERVAL{\gamma}$. In the first two cases, we know that the longest prefix of $S$ that matches $\ell(\gamma)$ has length smaller than $2^i$. Then, we read (but don't explicitly store) the label of $v$ in linear time as described in Theorem \ref{thm:verification}, finding the position of the leftmost mismatch with $S[1..2^i]$, if any. 
\end{proof}

Lemma \ref{lemma:leftToRight} is all we need to implement matching statistics with the CDAWG:

\begin{theorem} \label{thm:ms}
There is a data structure that takes $O(\newe_T)$ words of space, and that allows one to compute $\MS_{S,T}$ in $O(|S|)$ time and in $O(\min\{\mu,h_T\})$ words of working space, where $\mu$ is the largest number in $\MS_{S,T}$.
\end{theorem} 
\begin{proof}
We fill array $\MS_{S,T}$ from left to right, by implementing with $\CDAWG_T$ the classical matching statistics algorithm based on suffix link and child operations on the suffix tree. Assume that we have computed $\MS_{S,T}[1..i]$ for some $i$. Let $c=S[i+\MS_{S,T}[i]]$ and let $U = S[i..i+\MS_{S,T}[i]-1] = VX$, where $V$ is the longest prefix of $U$ that is right-maximal in $T$, and $v$ is the node of $\ST_T$ with label $V$. Assume that we know $v$ and the node $v'$ of $\CDAWG_T$ that corresponds to the equivalence class of $v$. Let $w'$ be the node of $\CDAWG_T$ that corresponds to the longest suffix of $\ell(v')$ that is a maximal repeat of $T$. If $|\ell(v)|>|\ell(w')|+1$, then $\MS_{S,T}[i+1]=\MS_{S,T}[i]-1$, since no suffix of $U$ longer than $|\ell(w')|+|X|$ can be followed by character $c$. Otherwise, we move to $w'$ in constant time by following the suffix pointer of $v'$, and we perform a blind search for $X$ from $w'$. Let $\ell(w')X=ZX'$, where $Z=\ell(z)$ is the longest prefix of $\ell(w')X$ that is right-maximal in $T$, and let $z'$ be the node of the CDAWG that corresponds to the equivalence class of $z$. If $|X'|>0$, or if no arc from $z'$ is labeled by $c$, then again $\MS_{S,T}[i+1]=\MS_{S,T}[i]-1$. Otherwise, we use Lemma \ref{lemma:leftToRight} to compute the length of the longest prefix of $S[i+\MS_{S,T}[i]..|S|]$ that matches a prefix of the arc from $z'$ labeled by $c$. The claimed time complexity comes from Lemma \ref{lemma:leftToRight} and from standard amortization arguments used in matching statistics. 
\end{proof}

Note that the data structure in Theorem \ref{thm:ms} takes actually $O(\min\{\newe_T,\newe_{\REV{T}}\})$ words of space, since one could index either $T$ or $\REV{T}$ for computing the matching statistics vector (in the latter case, $S$ is read from right to left). 

Another consequence of Property \ref{obs:grammar} is that we can compute the minimal absent words of $T$ using an index of size proportional just to the number of maximal repeats of $T$ and of their extensions:

\begin{lemma} \label{lemma:maws}
There is a data structure that takes $O(\newe_T+\newe_{\REV{T}})$ words of space, and that allows one to compute the minimal absent words of $T$ in $O(\newe_T+\newe_{\REV{T}}+\mathtt{out})$ time and in $O(\lambda_T + \min\{\mu_T,h_T\})$ words of working space, where $\mathtt{out}$ is the size of the output, $\lambda_T$ is the maximum number of left extensions of a maximal repeat of $T$, and $\mu_T$ is the length of a longest maximal repeat of $T$.
\end{lemma}
\begin{proof}
For every arc $\gamma=(v',w')$ of $\CDAWG_T$, we store in a variable $\gamma.\mathtt{order}$ the order of $v'$ among the in-neighbors of $w'$ induced by Property \ref{obs:inNeighbors} and used in $\REV{\CDAWG}_T$ (see Section \ref{sec:cdawg}), and we store in a variable $\gamma.\mathtt{previousChar}$ the character $a$, if any, such that $a\ell(v')b$ is a substring of $\ell(w')$ and $b=\gamma.\mathtt{char}$ is the first character of $\ell(\gamma)$.

Then, we traverse every node $v'$ of $\CDAWG_T$, and we scan every arc $\gamma=(v',w')$. If $\gamma.\mathtt{order}>1$, then $\ell(v')b$, where $b=\gamma.\mathtt{char}$, is always preceded by $\gamma.\mathtt{previousChar}$ in $T$, thus we print $a\ell(v')b$ to the output for all $a$ that label explicit and implicit Weiner links from $v'$ and that are different from $\gamma.\mathtt{previousChar}$. If $\gamma.\mathtt{order}=1$ then $\ell(v')b$ is a left-maximal substring of $T$, so we subtract the set of all Weiner links of $w'$ from the set of all Weiner links of $v'$ by a linear scan of their sorted lists, and we print $a\ell(v')b$ to the output for all characters $a$ in the resulting list. Note that the same Weiner link of $v'$ could be read multiple times, for multiple out-neighbors $w'$ of $v'$. However, every such access can be charged either to the output or to a corresponding Weiner link from $w'$, and each $w'$ takes part in at most one such subtraction. It follows that the time taken by all list subtractions is $O(\newe_{\REV{T}}+\mathtt{out})$.

We reconstruct each $\ell(v')$ in linear time as described in Theorem \ref{thm:verification}. 
\end{proof}

\section*{Acknowledgements}

We thank the anonymous reviewers for simplifying some parts of the paper, for improving its overall clarity, and for suggesting references \cite{gasieniec2005real,gasieniec2003time,lohrey2016traversing} and the current version of Lemma \ref{lemma:maws}.

\bibliographystyle{plain}
\bibliography{spire2017}

%

\end{document}